\documentclass[journal]{IEEEtran}

\usepackage[utf8]{inputenc}
\usepackage{color}
\usepackage{xcolor}
\usepackage{array}
\usepackage{verbatim}
\usepackage{float}
\usepackage{amsmath}
\usepackage{amsthm}
\usepackage{amssymb}
\usepackage{graphicx}
\usepackage{longtable}
\usepackage{multirow}
\usepackage[unicode=true,
bookmarks=false,
breaklinks=false,pdfborder={0 0 1},colorlinks=false]
{hyperref}
\hypersetup{
	colorlinks,bookmarksopen,bookmarksnumbered,citecolor=blue,urlcolor=blue}
\usepackage{cite}

\usepackage{lipsum}
\usepackage{mathtools}
\usepackage{cuted}

\floatstyle{ruled}
\newfloat{algorithm}{tbp}{loa}
\providecommand{\algorithmname}{Algorithm}
\floatname{algorithm}{\protect\algorithmname}

\makeatletter
\let\oldforeign@language\foreign@language
\DeclareRobustCommand{\foreign@language}[1]{%
	\lowercase{\oldforeign@language{#1}}}

\let\oldforeign@language\foreign@language
\DeclareRobustCommand{\foreign@language}[1]{%
	\lowercase{\oldforeign@language{#1}}}

\ifCLASSINFOpdf
\else
\fi

\newcommand{\MYfooter}{\smash{
		\hfil\parbox[t][\height][t]{\textwidth}{\centering
			\thepage}\hfil\hbox{}}}

\def\ps@IEEEtitlepagestyle{%
	\def\@oddhead{\parbox[t][\height][t]{\textwidth}{\centering \scriptsize
			Personal use of this material is permitted. Permission from the author(s) and/or copyright holder(s), must be obtained for all other uses. Please contact us and provide details if you believe this document breaches copyrights.\\
			\noindent\makebox[\linewidth]{}
		}\hfil\hbox{}}%
	\def\@evenhead{\scriptsize\thepage \hfil \leftmark\mbox{}}%
	\def\@oddfoot{\parbox[t][\height][l]{\textwidth}{
			\vspace{-20pt}{\rule{\textwidth}{0.4pt}}\\ \footnotesize\underline{To cite this article:}
			{\bf{\footnotesize\textcolor{red}{H. A. Hashim, "Asymptotically Stable Observer-based Controller for Attitude Tracking with Systematic Convergence," In Proc. of the IEEE 25th International Conference on System Theory, Control and Computing (ICSTCC), Iasi, Romania, 2021, pp. 39-45.}}} doi: \href{https://doi.org/10.1109/ICSTCC52150.2021.9607198}{10.1109/ICSTCC52150.2021.9607198}\\
			\noindent\makebox[\linewidth]
		}\hfil\hbox{}}%
	\def\@evenfoot{\MYfooter}}

\makeatother
\newtheorem{thm}{Theorem}
\newtheorem{rem}{Remark}

\newtheorem{assum}{Assumption}

\begin{document}
	\bstctlcite{IEEEexample:BSTcontrol}

	\title{Asymptotically Stable Observer-based Controller for Attitude Tracking with Systematic Convergence}

	\author{Hashim A. Hashim
		\thanks{This work was supported in part by Thompson Rivers University Internal research fund \# 102315.}
	\thanks{$^*$Corresponding author, H. A. Hashim is with the Department of Engineering and Applied Science, Thompson Rivers University, Kamloops, British Columbia, Canada, V2C-0C8, e-mail: hhashim@tru.ca.}}

	
	
	\maketitle
	
	\begin{abstract}
This paper proposes a novel unit-quaternion observer-based controller
for attitude tracking (attitude and angular velocity) with guaranteed
transient and steady-state performance. The proposed approach is computationally
cheap and can operate based on measurements provided, for instance by a typical
low-cost inertial measurement unit (IMU) or magnetic, angular rate, and gravity (MARG) sensor without the knowledge of angular velocity. First, an observer evolved on $\mathbb{S}^{3}\times\mathbb{R}^{3}$
is developed guaranteeing asymptotic stability of the closed loop
error signals starting from any initial condition. Afterwords, the
observer is combined with the proposed controller such that the observer-based
controller ensures asymptotic stability of the closed loop error signals
starting from any initial condition. Simulation performed in discrete
form at low sampling rate reveals the robustness and effectiveness
of the proposed approach.
	\end{abstract}
	

	\IEEEpeerreviewmaketitle{}

	\section{Introduction}
	
	\IEEEPARstart{A}{ttitude} tracking is a fundamental part of a variety of robotics applications
	including space telescopes, unmanned aerial vehicles, rotating radars
	and others. Development of cheap, small-sized, low-weight, and power-efficient
	inertial measurement units (IMUs) sparked a wave of active research
	in the area of attitude observation and tracking (observer + control)
	\cite{salcudean1991globally,hashim2020SO3Wiley,lizarralde1996attitude,mayhew2011quaternion,hashim2021_COMP_ENG_PRAC,hashim2018SO3Stochastic,castillo2019disturbance,Hashim2021AESCTE,cui2019velocity,hashim2021ACC,hashim2021T_SMCS_SLAM}.
	The main challenge of working with low-cost sensors, such as IMUs
	is their susceptibility to noise. Also, the true attitude dynamics
	rely on angular velocity of a rigid-body, commonly measured by a gyroscope.
	Nonetheless, replacement of a failed gyroscope often proves to be
	challenging and costly \cite{Nasa2018Hubble}. Hence, there is a need
	for effective attitude tracking solutions that do not require knowledge
	of angular velocity. 
	
	Velocity-free attitude control is possible with a full-state observer
	able to provide accurate estimates of both attitude and angular velocity.
	Thereafter, an attitude tracking control based solely on available
	attitude and angular velocity estimates is developed. It is worth
	mentioning that angular velocity is observable only given the knowledge
	of attitude. In turn, in order to obtain attitude information it is
	sufficient to acquire at least two vectorial measurements at the rigid-body
	using, for example, an IMU module \cite{hashim2020SO3Wiley,shuster1981three,hashim2018SO3Stochastic,odry2018kalman,odry2021open}.
	An early solution presented a full state observer for rigid-body motion
	\cite{salcudean1991globally}. The problem of velocity-free attitude
	tracking has been addressed in the literature in a variety of ways
	including a full-state observer for rigid-body motion \cite{salcudean1991globally},
	an observer-based controller with local exponential stability \cite{caccavale1999output},
	attitude tracking control without velocity measurements \cite{costic2001quaternion},
	and a hybrid attitude tracking controller with semi-global asymptotic
	stability which has a switching observer that restores angular velocity
	signal \cite{mayhew2011quaternion}. Additionally, an observer-based
	controller for the attitude tracking problem has been formulated to
	handle unknown bounded external disturbances \cite{xia2010attitude,tognon2017dynamics}.
	
	Despite a multitude of existing solutions, the common issue all of
	them share is the inability to guarantee the attitude error transient
	and steady-state performance. Lack of certainty and predictability
	in the performance of attitude observation and the control errors
	can easily destabilize the entire process. However, full control over
	the transient and steady-state performance can be gained by employing
	a prescribed performance function (PPF) \cite{bechlioulis2008robust}.
	PPF is able to guide the error to initiate within a large set and
	reduce systematically to settle within a small set. PPF approach has
	been successfully utilized for attitude-related problems, for example,
	an observer-based controller for attitude tracking problem subject
	to actuator saturation \cite{zhang2019observer} and output feedback
	of attitude problem subject to external disturbances \cite{yin2019appointed}.
	The work in \cite{zhang2019observer} used angle-axis, whereas the
	work in \cite{yin2019appointed} considered Rodriguez parameters for
	attitude parameterization. Both angle-axis and Rodriguez parameters
	approaches for attitude parameterization are subject to singularity.
	Moreover, the overall closed loop signals of \cite{zhang2019observer,yin2019appointed}
	are shown to be semi-globally uniformly ultimately bounded, and therefore,
	the asymptotic stability cannot be guaranteed. 
	
	Considering the above literature overview, it becomes apparent that
	in order to achieve a stable attitude tracking process and alleviate
	the need for angular velocity information, observer-based control
	solutions with guaranteed measures of transient and steady-state performance
	of attitude error should be developed. Thus, the main contributions
	of this work are: 1) an attitude and angular velocity observer developed
	on $\mathbb{S}^{3}\times\mathbb{R}^{3}$ guaranteeing almost global
	asymptotic stability with predefined measures of transient and steady-state
	performance of attitude error is proposed. 2) The estimates of attitude
	and angular velocity obtained by the observer are combined with a
	novel attitude tracking control law that ensures almost global asymptotic
	stability with guaranteed measures of transient and steady-state performance
	of attitude error. 3) The proposed solutions produce accurate results
	even when supplied with uncertain measurements obtained from a low-cost
	IMU module at low sampling rate.
	
	The rest of the article is organized as follows: Section \ref{sec:Attit_Problem-Formulation}
	introduces the math notation, unit-quaternion preliminaries, attitude
	dynamics, available measurements, and attitude error. Section \ref{sec:Systematic-Performance}
	presents the concept of PPF. Section \ref{sec:Attit_Observer_Controller}
	introduces the observer-based controller for the attitude tracking
	problem. Section \ref{sec:SE3_Simulations} demonstrates the robustness
	of the proposed approach through numerical results. Finally, Section
	\ref{sec:SE3_Conclusion} summarizes the work.
	
	\section{Problem Formulation\label{sec:Attit_Problem-Formulation}}
	
	\subsection{Preliminaries}
	
	Let $\mathbb{R}$ and $\mathbb{R}^{n\times m}$ denote a set of real
	numbers and a real $n$-by-$m$ dimensional space, respectively. For
	$x\in\mathbb{R}^{n}$, $||x||=\sqrt{x^{\top}x}$ denotes the Euclidean
	norm. $\left\{ \mathcal{I}\right\} $ and $\{\mathcal{B}\}$ correspond
	to fixed inertial-frame and body-frame, respectively. $R\in\mathbb{SO}\left(3\right)$
	describes rigid-body's attitude (orientation) where \cite{hashim2018SO3Stochastic,hashim2020SO3Wiley}
	\[
	\mathbb{SO}(3)=\{R\in\mathbb{R}^{3\times3}|RR^{\top}=R^{\top}R=\mathbf{I}_{3}\text{, }{\rm det}(R)=+1\}
	\]
	with ${\rm det}(\cdot)$ being a determinant. Let $Q=[q_{0},q^{\top}]^{\top}\in\mathbb{S}^{3}$
	stand for a unit-quaternion vector where $q_{0}\in\mathbb{R}$ and
	$q\in\mathbb{R}^{3}$ such that $\mathbb{S}^{3}=\{\left.Q\in\mathbb{R}^{4}\right|||Q||=\sqrt{q_{0}^{2}+q^{\top}q}=1\}$.
	$[\Omega]_{\times}$ represents a skew symmetric matrix {\small{}with
		\begin{align*}
		\left[\Omega\right]_{\times} & =\left[\begin{array}{ccc}
		0 & -\Omega_{3} & \Omega_{2}\\
		\Omega_{3} & 0 & -\Omega_{1}\\
		-\Omega_{2} & \Omega_{1} & 0
		\end{array}\right],\hspace{1em}\Omega=\left[\begin{array}{c}
		\Omega_{1}\\
		\Omega_{2}\\
		\Omega_{3}
		\end{array}\right]
		\end{align*}
	}Note that $[\Omega]_{\times}y=\Omega\times y$ for all $\Omega,y\in\mathbb{R}^{3}$.
	Let $Q^{-1}=[q_{0},-q^{\top}]^{\top}\in\mathbb{S}^{3}$ be the inverse
	of $Q$. For $Q_{1}=[\begin{array}{cc}
	q_{01} & q_{1}^{\top}\end{array}]^{\top}\in\mathbb{S}^{3}$ and $Q_{2}=[\begin{array}{cc}
	q_{02} & q_{2}^{\top}\end{array}]^{\top}\in\mathbb{S}^{3}$, the quaternion product is as follows:
	\[
	Q_{1}\odot Q_{2}=\left[\begin{array}{c}
	q_{01}q_{02}-q_{1}^{\top}q_{2}\\
	q_{01}q_{2}+q_{02}q_{1}+[q_{1}]_{\times}q_{2}
	\end{array}\right]\in\mathbb{S}^{3}
	\]
	For $Q=[q_{0},q^{\top}]^{\top}\in\mathbb{S}^{3}$, the related unit-quaternion
	mapping from $\mathbb{S}^{3}$ to $\mathbb{SO}\left(3\right)$ is
	as follows:
	\begin{align}
	\mathcal{R}_{Q} & =(q_{0}^{2}-||q||^{2})\mathbf{I}_{3}+2qq^{\top}+2q_{0}\left[q\right]_{\times}\in\mathbb{SO}\left(3\right)\label{eq:Attit_map_Q_SO3}
	\end{align}
	In view of \eqref{eq:Attit_map_Q_SO3}, the quaternion identity and
	its related mapping to $\mathbb{SO}\left(3\right)$ are defined by
	\begin{equation}
	\mathbf{Q}_{{\rm I}}=[\pm1,0,0,0]^{\top}\hspace{1em}\Leftrightarrow\hspace{1em}\mathcal{R}_{\mathbf{Q}_{{\rm I}}}=\mathbf{I}_{3}\label{eq:Attit_Ident_Q_SO3}
	\end{equation}
	For more details, visit \cite{hashim2019AtiitudeSurvey,shuster1993survey}.
	For any $\Omega\in\mathbb{R}^{3}$ and $Q\in\mathbb{S}^{3}$, the
	following maps are considered:
	\begin{equation}
	\begin{cases}
	\overline{\Omega} & =[0,\Omega^{\top}]^{\top}\in\mathbb{R}^{4}\\
	\Gamma(\Omega) & =\left[\begin{array}{cc}
	0 & -\Omega^{\top}\\
	\Omega & -\left[\Omega\right]_{\times}
	\end{array}\right]\in\mathbb{R}^{4\times4}
	\end{cases}\label{eq:Attit_Map1}
	\end{equation}

	\subsection{Measurements and Dynamics}
	
	Let $Q\in\mathbb{S}^{3}$ and $\Omega\in\mathbb{R}^{3}$ denote the
	true unit-quaternion and angular velocity of a rigid-body in 3D space,
	respectively as depicted in Fig. \ref{fig:SO3_PPF_STCH_1}.
	\begin{figure}[h]
		\centering{}\includegraphics[scale=0.33]{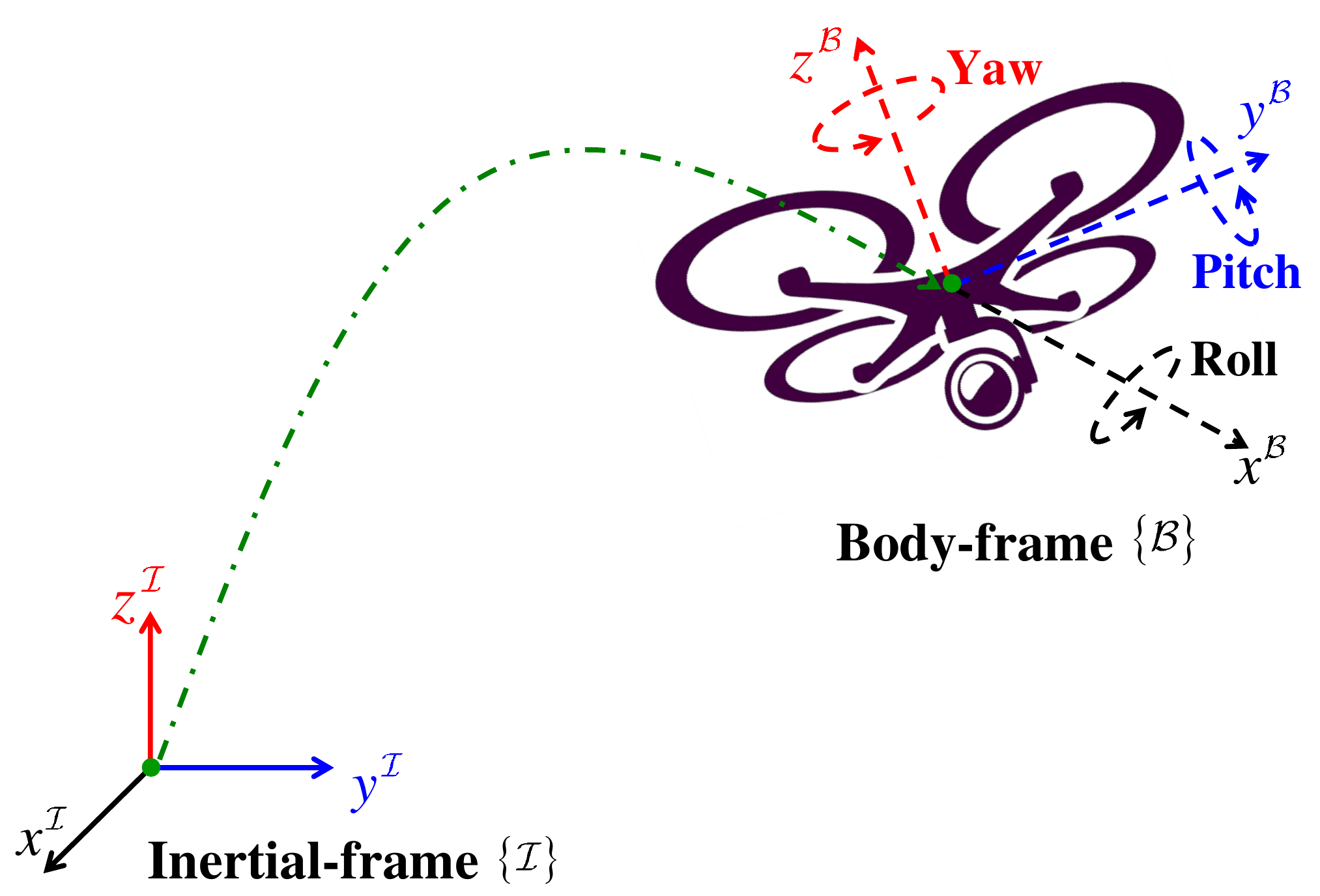}\caption{\label{fig:SO3_PPF_STCH_1} Rigid-body's orientation (attitude) in
			body-frame $\{\mathcal{B}\}$ relative to inertial-frame $\left\{ \mathcal{I}\right\} $.}
	\end{figure}
	The true attitude dynamics on $\mathbb{SO}(3)$ are
	\begin{equation}
	\dot{R}=R\left[\Omega\right]_{\times},\hspace{1em}J\dot{\Omega}=\left[J\Omega\right]_{\times}\Omega+\tau\label{eq:Attit_True_Dyn_SO3}
	\end{equation}
	The equivalent unit-quaternion representation is as follows:
	\begin{equation}
	\begin{cases}
	\dot{Q} & =\frac{1}{2}Q\odot\overline{\Omega}=\frac{1}{2}\Gamma(\Omega)Q\\
	J\dot{\Omega} & =\left[J\Omega\right]_{\times}\Omega+\tau
	\end{cases}\label{eq:Attit_True_Dyn_Q}
	\end{equation}
	with $J$ being a positive-definite symmetric inertia matrix of a
	rigid-body, and $\tau\in\mathbb{R}^{3}$ being input torque. Note
	that $R,Q,\Omega,\tau,J\in\{\mathcal{B}\}$. A typical low-cost IMU
	module collects measurements using a 3-axis magnetometer and a 3-axis
	accelerometer at the body-frame $\{\mathcal{B}\}$ which can be expressed
	as follows:
	\begin{equation}
	\overline{b_{i}}=Q^{-1}\odot\overline{r_{i}}\odot Q+\overline{n_{i}}\in\mathbb{R}^{4},\hspace{1em}\forall i=1,2,\ldots,n\label{eq:Attit_Vect_True}
	\end{equation}
	where $\overline{r_{i}}=[0,r_{i}^{\top}]^{\top}\in\mathbb{R}^{4}$
	is the $i$th known observation in the inertial frame $\{\mathcal{I}\}$
	and $n_{i}$ is unknown noise. The vectors in \eqref{eq:Attit_Vect_True}
	are normalized as follows:
	\begin{equation}
	\mathfrak{r}_{i}=r_{i}/||r_{i}||,\hspace{1em}\mathfrak{b}_{i}=b_{i}/||b_{i}||\label{eq:Attit_Vector_norm}
	\end{equation}
	Let $M^{\mathcal{I}}=[\mathfrak{r}_{1},\mathfrak{r}_{2},\ldots,\mathfrak{r}_{n}]$
	and $M^{\mathcal{B}}=[\mathfrak{b}_{1},\mathfrak{b}_{2},\ldots,\mathfrak{b}_{n}]$.
	\begin{rem}
		\label{rem:Rem1_Attit_Rec}\cite{hashim2018SO3Stochastic,hashim2020SO3Wiley}
		The attitude can be reconstructed if (${\rm rank}(M^{\mathcal{I}})={\rm rank}(M^{\mathcal{B}})=3$).
		For $n=2$, the third vector is defined by $\mathfrak{r}_{3}=\mathfrak{r}_{1}\times\mathfrak{r}_{2}$
		and $\mathfrak{b}_{3}=\mathfrak{b}_{1}\times\mathfrak{b}_{2}$.
	\end{rem}
	Let $Q_{y}\in\mathbb{S}^{3}$ denote a reconstructed quaternion of
	the true quaternion $Q\in\mathbb{S}^{3}$. $Q_{y}$ can be obtained
	by one of the methods of quaternion determination, for instance QUEST
	\cite{shuster1981three}, optimal QUEST \cite{mortari2000second},
	or for others see \cite{hashim2020AtiitudeSurvey}. The desired (reference)
	trajectory of $Q_{d}\in\mathbb{S}^{3}$ is defined by the desired
	angular velocity $\Omega_{d}\in\mathbb{R}^{3}$ as follows:
	\begin{align}
	\dot{Q}_{d} & =\frac{1}{2}Q_{d}\odot\overline{\Omega_{d}}=\frac{1}{2}\Gamma(\Omega_{d})Q_{d},\hspace{1em}Q_{d}(0)\in\mathbb{S}^{3}\label{eq:Attit_Rdot_desired}
	\end{align}
	
	\begin{assum}\label{Assum:Attit_desired_w}Both $\Omega_{d}$ and
		$\dot{\Omega}_{d}$ are upper bounded by a scalar $C_{d}<\infty$
		with $C_{d}\geq\max\{\sup_{t\geq0}||\Omega_{d}||,\sup_{t\geq0}||\dot{\Omega}_{d}||\}$.\end{assum}
	
	Recall that the aim of this work is to design an observer-based controller
	characterized by guaranteed measures of transient and steady-state
	performance that does not require knowledge of angular velocity. Therefore,
	the first step consists in designing a full-state observer (attitude
	and angular velocity) evolved on $\mathbb{S}^{3}\times\mathbb{R}^{3}$
	ensuring almost global asymptotic stability and following predefined
	measures of transient and steady-state performance of attitude observation
	error. Next, the full-state observer is combined with the attitude
	tracking control on $\mathbb{S}^{3}$ ensuring almost global asymptotic
	stability with guaranteed performance of attitude tracking error.
	The angular velocity is observable if the attitude is known. In view
	of Remark \ref{rem:Rem1_Attit_Rec}, the attitude, in turn, is observable
	if ${\rm rank}(M^{\mathcal{I}})={\rm rank}(M^{\mathcal{B}})=3$. Let
	$\hat{Q}=[\hat{q}_{0},\hat{q}^{\top}]^{\top}\in\mathbb{S}^{3}$, be
	the estimate of $Q=[q_{0},q^{\top}]^{\top}\in\mathbb{S}^{3}$ and
	define the error between $Q$ and $\hat{Q}$ as
	\begin{align}
	\tilde{Q}_{o} & =[\tilde{q}_{o0},\tilde{q}_{o}^{\top}]^{\top}=\hat{Q}^{-1}\odot Q\in\mathbb{S}^{3}\label{eq:Attit_Ro_error}
	\end{align}
	Define $\hat{\Omega}\in\mathbb{R}^{3}$ as the estimate of $\Omega$
	and let the error between $\Omega$ and $\hat{\Omega}$ be
	\begin{equation}
	\overline{\tilde{\Omega}_{o}}=\overline{\Omega}-\tilde{Q}_{o}^{-1}\odot\overline{\hat{\Omega}}\odot\tilde{Q}_{o}\hspace{0.5em}\Leftrightarrow\hspace{0.5em}\tilde{\Omega}_{o}=\Omega-\mathcal{R}_{\tilde{Q}_{o}}^{\top}\hat{\Omega}\in\mathbb{R}^{3}\label{eq:Attit_Omo_error}
	\end{equation}
	where $\mathcal{R}_{\tilde{Q}_{o}}=(\tilde{q}_{o0}^{2}-||\tilde{q}_{o}||^{2})\mathbf{I}_{3}+2\tilde{q}_{o}\tilde{q}_{o}^{\top}+2\tilde{q}_{o0}[\tilde{q}_{o}]_{\times}$,
	see the map in \eqref{eq:Attit_map_Q_SO3}. Assume that $Q_{d}=[q_{d0},q_{d}^{\top}]^{\top}\in\mathbb{S}^{3}$
	is the desired quaternion trajectory, and let the error between $Q$
	and $Q_{d}$ be
	\begin{equation}
	\tilde{Q}_{c}=[\tilde{q}_{c0},\tilde{q}_{c}^{\top}]^{\top}=Q_{d}^{-1}\odot Q\in\mathbb{S}^{3}\label{eq:Attit_Rc_error}
	\end{equation}
	Allow $\Omega_{d}\in\mathbb{R}^{3}$ to be the desired trajectory
	of angular velocity, and let the error between $\Omega$ and $\Omega_{d}$
	be
	\begin{equation}
	\overline{\tilde{\Omega}_{c}}=\overline{\Omega}-\tilde{Q}_{c}^{-1}\odot\overline{\Omega_{d}}\odot\tilde{Q}_{c}\hspace{0.5em}\Leftrightarrow\hspace{0.5em}\tilde{\Omega}_{c}=\Omega-\mathcal{R}_{\tilde{Q}_{c}}^{\top}\Omega_{d}\in\mathbb{R}^{3}\label{eq:Attit_Omc_error}
	\end{equation}
	where $\mathcal{R}_{\tilde{Q}_{c}}=(\tilde{q}_{c0}^{2}-||\tilde{q}_{c}||^{2})\mathbf{I}_{3}+2\tilde{q}_{c}\tilde{q}_{c}^{\top}+2\tilde{q}_{c0}[\tilde{q}_{c}]_{\times}$.
	Recall \eqref{eq:Attit_Ident_Q_SO3}, the objective of attitude observation
	is to drive $\tilde{Q}_{o}\rightarrow\mathbf{Q}_{{\rm I}}$ and $\tilde{\Omega}_{o}\rightarrow0_{3\times1}$.
	Similarly, the objective of attitude control is to drive $\tilde{Q}_{c}\rightarrow\mathbf{Q}_{{\rm I}}$
	and $\tilde{\Omega}_{c}\rightarrow0_{3\times1}$, which, in turn,
	implies $\mathcal{R}_{\tilde{Q}_{o}}\rightarrow\mathbf{I}_{3}$ and
	$\mathcal{R}_{\tilde{Q}_{c}}\rightarrow\mathbf{I}_{3}$, visit \eqref{eq:Attit_Ident_Q_SO3}.
	By the definition of unit-quaternion and the identity property, $\tilde{q}_{o0}\rightarrow\pm1$
	implies that $\tilde{q}_{o}\rightarrow0_{3\times1}$ and vice versa.
	Likewise, $\tilde{q}_{c0}\rightarrow\pm1$ indicates that $\tilde{q}_{c}\rightarrow0_{3\times1}$
	and vice versa.
	
	\section{Guaranteed Performance\label{sec:Systematic-Performance}}
	
	This section aims to guarantee that the tracking performance of
	\begin{equation}
	e_{\star}=1-\tilde{q}_{\star0}\label{eq:Attit_Pres_error}
	\end{equation}
	is initiated within a known large set and decreased smoothly to stay
	within a known small set where the subscript $\star$ is to be replaced
	by $o$ and $c$. Note that unit-quaternion is subject to non-uniqueness
	such that for $\tilde{Q}_{\star}=-\tilde{Q}_{\star}\in\mathbb{S}^{3}$
	one has $\mathcal{R}_{\tilde{Q}_{\star}}\in\mathbb{SO}\left(3\right)$.
	As such, in the algorithm setup it is not hard to obtain $\tilde{q}_{\star0}\in\mathbb{R}_{+}$
	for all $t\geq0$. Define the following positive time-decreasing prescribed
	performance function (PPF) with the map $\xi_{\star}:\mathbb{R}_{+}\to\mathbb{R}_{+}$
	\cite{bechlioulis2008robust}
	\begin{equation}
	\xi_{\star}(t)=(\xi_{\star}^{0}-\xi_{\star}^{\infty})\exp(-\ell_{\star}t)+\xi_{\star}^{\infty}\label{eq:Attit_Presc}
	\end{equation}
	with $\xi_{\star}(0)=\xi_{\star}^{0}>0$ and $\xi_{\star}^{\infty}>0$
	being the upper bounds of a known large set and small set, respectively,
	and $\ell_{\star}>0$ being the convergence rate of $\xi_{\star}=\xi_{\star}(t)$
	from $\xi_{\star}^{0}$ to $\xi_{\star}^{\infty}$. It can be deduced
	that $\lim\limits _{t\to\infty}\xi_{\star}=\xi_{\star}^{\infty}$.
	$e_{\star}=e_{\star}(t)$ can be controlled by the predefined transient
	and steady-state boundaries provided that
	\begin{align}
	-\underline{\delta}_{\star}\xi_{\star}<e_{\star}<\xi_{\star}, & \text{ if }e_{\star}(0)\geq0\label{eq:Attit_ePos}
	\end{align}
	where $\underline{\delta}_{\star}\in[0,1]$. Due to the fact that
	$e_{\star}\in[0,1]\forall t\geq0$, $e_{\star}$ is controlled by
	the PPF if the condition in \eqref{eq:Attit_ePos} is met. Fig. \ref{fig:NAVPPF_2}
	illustrates the concept of PPF in action allowing for the desired
	convergence of the constrained error $e_{\star}$ in \eqref{eq:Attit_ePos}.
	
	\begin{figure}[H]
		\centering{}\includegraphics[scale=0.32]{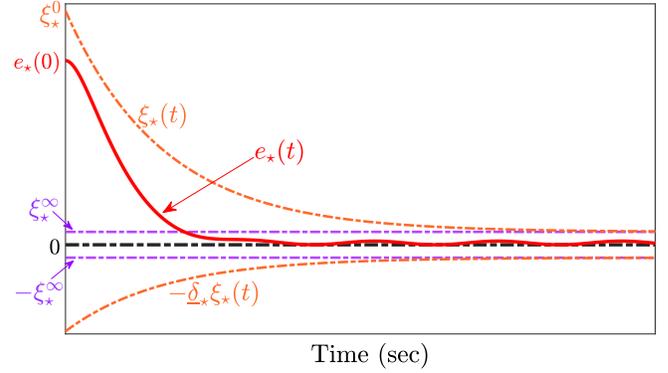} \caption{Convergence of $e_{\star}(t)$ controlled by the PPF according to
			Eq. \eqref{eq:Attit_ePos}.}
		\label{fig:NAVPPF_2}
	\end{figure}
	
	From \eqref{eq:Attit_ePos} and Fig. \ref{fig:NAVPPF_2}, $e_{\star}$
	is constrained. Let us redefine the constrained error $e_{\star}$
	as
	\begin{equation}
	e_{\star}=\xi_{\star}\mathcal{N}(E_{\star})\label{eq:Attit_e_Constrained}
	\end{equation}
	where $\xi_{\star}$ is given in \eqref{eq:Attit_Presc}, $E_{\star}\in\mathbb{R}$
	is transformed (unconstrained) error, and $\mathcal{N}(E_{\star})$
	follows Assumption \ref{Assum:Attit_2PPF}:
	
	\begin{assum}\label{Assum:Attit_2PPF}$\mathcal{N}(E_{\star})$ is
		characterized by the following properties \cite{bechlioulis2008robust}: 
		\begin{enumerate}
			\item[\textbf{1.}] $\mathcal{N}(E_{\star})$ is smooth and strictly increasing. 
			\item[\textbf{2.}] $\mathcal{N}(E_{\star})$ is constrained by $-\underline{\delta}_{\star}<\mathcal{N}(E_{\star})<\overline{\delta}_{\star}$
			with $\overline{\delta}_{\star}$ and $\underline{\delta}_{\star}$
			being positive constants and $\underline{\delta}_{\star}\leq\overline{\delta}_{\star}$. 
			\item[\textbf{3.}] $\underset{E_{i}\rightarrow-\infty}{\lim}\mathcal{N}(E_{\star})=-\underline{\delta}_{i}$
			and $\underset{E_{i}\rightarrow+\infty}{\lim}\mathcal{N}(E_{\star})=\bar{\delta}_{i}$
			where
		\end{enumerate}
		\begin{equation}
		\mathcal{N}(E_{\star})=\frac{\overline{\delta}_{\star}\exp(E_{\star})-\underline{\delta}_{\star}\exp(-E_{\star})}{\exp(E_{\star})+\exp(-E_{\star})}\label{eq:Attit_Smooth}
		\end{equation}
	\end{assum}
	
	Based on \eqref{eq:Attit_e_Constrained} one finds
	\begin{equation}
	\begin{aligned}E_{\star}= & \mathcal{N}^{-1}(e_{\star}/\xi_{\star})=\frac{1}{2}\text{ln}\frac{\underline{\delta}_{\star}+e_{\star}/\xi_{\star}}{\overline{\delta}_{\star}-e_{\star}/\xi_{\star}},\hspace{0.2em}\overline{\delta}_{\star}\geq\underline{\delta}_{\star}\end{aligned}
	\label{eq:Attit_trans}
	\end{equation}
	
	\begin{rem}
		\label{rem:Attit_PFF_e_E_0}\cite{hashim2020TITS_SLAM,hashim2020SO3Wiley}
		It becomes apparent that selecting $\underline{\delta}_{\star}=\overline{\delta}_{\star}$
		implies that $E_{\star}>0$ for all $e_{\star}>0$, and $E_{\star}=0$
		only at $e_{\star}=0$. Thus, the critical point of $e_{\star}$ coincides
		with the critical point of $E_{\star}$.
	\end{rem}
	\begin{rem}
		\label{rem:Attit_PPF1}\cite{hashim2020TITS_SLAM,hashim2020SO3Wiley}
		From \eqref{eq:Attit_Presc}, \eqref{eq:Attit_e_Constrained}, and
		\eqref{eq:Attit_trans}, $e_{\star}$ is constrained by $\xi_{\star}$
		if and only if $E_{\star}\in\mathcal{L}_{\infty}$.
	\end{rem}
	Define
	\begin{equation}
	\begin{split}\Delta_{\star} & =\frac{1}{2\xi_{\star}}\frac{\partial\mathcal{N}^{-1}(e_{\star}/\xi_{\star})}{\partial(e_{\star}/\xi_{\star})}\\
	& =\frac{1/(2\xi_{\star})}{\underline{\delta}_{\star}+e_{\star}/\xi_{\star}}+\frac{1/(2\xi_{\star})}{\overline{\delta}_{\star}-e_{\star}/\xi_{\star}}
	\end{split}
	\label{eq:Attit_mu}
	\end{equation}
	where $\Delta_{\star}$ is a positive function. Hence, one obtains
	\begin{align}
	\dot{E}_{\star} & =\Delta_{\star}(\dot{e}_{\star}-\frac{\dot{\xi}_{\star}}{\xi_{\star}}e_{\star})\label{eq:Attit_Trans_dot}
	\end{align}

	\section{Observer-based Controller with Guaranteed Convergence\label{sec:Attit_Observer_Controller}}
	
	\subsection{Full State Observer with Guaranteed Performance\label{subsec:Observer}}
	
	Let $\hat{Q}\in\mathbb{S}^{3}$ and $\hat{\Omega}\in\mathbb{R}^{3}$
	denote the estimates of $Q$ and $\Omega$, respectively. Consider
	the following attitude and angular velocity observer design:
	\begin{equation}
	\begin{cases}
	\dot{\hat{Q}} & =\frac{1}{2}\hat{Q}\odot\overline{\hat{\Omega}+W_{\Omega}}=\frac{1}{2}\Gamma(\hat{\Omega}+W_{\Omega})\hat{Q}\\
	\hat{J}\dot{\hat{\Omega}} & =[\hat{J}\hat{\Omega}]_{\times}\hat{\Omega}+\hat{\tau}+\hat{J}[\hat{\Omega}]_{\times}W_{\Omega}+W_{\tau}\\
	W_{\Omega} & =-k_{o}(E_{o}\Delta_{o}+1)\mathcal{R}_{\tilde{Q}_{o}}\tilde{q}_{o}\\
	W_{\tau} & =-\gamma_{o}(E_{o}\Delta_{o}+1)\mathcal{R}_{\tilde{Q}_{o}}\tilde{q}_{o}
	\end{cases}\label{eq:Attit_Obsv_Design}
	\end{equation}
	where $\tilde{Q}_{o}=[\tilde{q}_{o0},\tilde{q}_{o}^{\top}]^{\top}=\hat{Q}^{-1}\odot Q_{y}$
	denotes the unit-quaternion error in observation, $Q_{y}$ stands
	for a reconstructed unit-quaternion obtained, for instance, by QUEST
	algorithm \cite{shuster1981three,hashim2020AtiitudeSurvey}, $\mathcal{R}_{\tilde{Q}_{o}}$
	is the attitude observation error, $[0,\hat{\tau}^{\top}]^{\top}=\tilde{Q}_{o}\odot\overline{\tau}\odot\tilde{Q}_{o}^{-1}$,
	or more simply, $\hat{\tau}=\mathcal{R}_{\tilde{Q}_{o}}\tau$ denotes
	the torque input described in the observer frame, $\hat{J}=\mathcal{R}_{\tilde{Q}_{o}}J\mathcal{R}_{\tilde{Q}_{o}}^{\top}$
	represents the inertia matrix described in the observer-frame, $e_{o}=1-\tilde{q}_{o0}$,
	and $E_{o}=\frac{1}{2}\text{ln}\frac{\underline{\delta}_{o}+e_{o}/\xi_{o}}{\overline{\delta}_{o}-e_{o}/\xi_{o}}$
	denotes the transformed error. Additionally, $\xi_{o}$ is the PPF
	defined in \eqref{eq:Attit_Presc} where $\xi_{o}^{0}>e_{o}(0)$,
	$W_{\Omega}$ and $W_{\tau}$ are correction factors, and $k_{o}$,
	$\gamma_{o}$, and $\underline{\delta}_{o}=\overline{\delta}_{o}>e_{o}(0)$
	are positive constants.
	\begin{thm}
		\label{thm:Theorem1_OBSV}Consider the dynamics in \eqref{eq:Attit_True_Dyn_Q}
		and the observer in \eqref{eq:Attit_Obsv_Design}. Let Assumption
		\ref{Assum:Attit_desired_w} hold true given that the condition in
		Remark \ref{rem:Rem1_Attit_Rec} is met. Let $k_{o}$, $\gamma_{o}$,
		$\underline{\delta}_{o}=\overline{\delta}_{o}>e_{o}(0)$, $\xi_{o}^{0}>e_{o}(0)$,
		and $\xi_{o}^{\infty}$ be positive constants. Then for $E_{o}(0)\in\mathcal{L}_{\infty}$,
		1) $E_{o}$, $e_{o}$, and $\hat{\Omega}$ are globally bounded, and
		2) starting from any initial conditions, all $E_{o}$, $e_{o}$, and
		$\tilde{\Omega}_{o}$ converge asymptotically to the origin with $\lim_{t\rightarrow\infty}\tilde{q}_{o}\rightarrow0_{3\times1}$
		and $\lim_{t\rightarrow\infty}\tilde{q}_{o0}\rightarrow\pm1$.
	\end{thm}
	\begin{proof}Recall the error in \eqref{eq:Attit_Ro_error}, $\tilde{Q}_{o}=\hat{Q}^{-1}\odot Q$.
		From \eqref{eq:Attit_True_Dyn_Q} and \eqref{eq:Attit_Obsv_Design},
		one obtains
		\begin{align}
		\dot{\tilde{Q}}_{o} & =\dot{\hat{Q}}^{-1}\odot Q+\hat{Q}^{-1}\odot\dot{Q}\nonumber \\
		& =-\frac{1}{2}\overline{\hat{\Omega}+W_{\Omega}}\odot\tilde{Q}_{o}+\frac{1}{2}\tilde{Q}_{o}\odot\overline{\Omega}\nonumber \\
		& =\frac{1}{2}\tilde{Q}_{o}\odot(\overline{\Omega}-\tilde{Q}_{o}^{-1}\odot\overline{\hat{\Omega}+W_{\Omega}}\odot\tilde{Q}_{o})\nonumber \\
		& =\frac{1}{2}\tilde{Q}_{o}\odot\left[\begin{array}{c}
		0\\
		\tilde{\Omega}_{o}-\mathcal{R}_{\tilde{Q}_{o}}^{\top}W_{\Omega}
		\end{array}\right]\label{eq:Attit_Qo_dot}
		\end{align}
		In view of \eqref{eq:Attit_True_Dyn_SO3} and \eqref{eq:Attit_True_Dyn_Q},
		the mapping of \eqref{eq:Attit_Qo_dot} to $\mathbb{SO}(3)$ is
		\begin{equation}
		\dot{\mathcal{R}}_{\tilde{Q}_{o}}=\mathcal{R}_{\tilde{Q}_{o}}[\tilde{\Omega}_{o}-\mathcal{R}_{\tilde{Q}_{o}}^{\top}W_{\Omega}]_{\times}\label{eq:Attit_Ro_dot}
		\end{equation}
		Hence, the dynamics in \eqref{eq:Attit_Qo_dot} become
		\begin{align}
		\left[\begin{array}{c}
		\dot{\tilde{q}}_{o0}\\
		\dot{\tilde{q}}_{o}
		\end{array}\right]= & \frac{1}{2}\left[\begin{array}{c}
		-\tilde{q}_{o}^{\top}\\
		\tilde{q}_{o0}\mathbf{I}_{3}+\left[\tilde{q}_{o}\right]_{\times}
		\end{array}\right](\tilde{\Omega}_{o}-\mathcal{R}_{\tilde{Q}_{o}}^{\top}W_{\Omega})\label{eq:Attit_RI2_dot}
		\end{align}
		As such, $\dot{E}_{o}$ in \eqref{eq:Attit_Trans_dot} is as follows:
		\begin{align}
		\dot{E}_{o} & =\Delta_{o}(-\frac{1}{2}\tilde{q}_{o}^{\top}(\tilde{\Omega}_{o}-\mathcal{R}_{\tilde{Q}_{o}}^{\top}W_{\Omega})-\frac{\dot{\xi}_{o}}{\xi_{o}}(1-\tilde{q}_{o0}))\label{eq:Attit_Eo_dot}
		\end{align}
		where $\Delta_{o}=\frac{1/(2\xi_{o})}{\underline{\delta}_{o}+e_{o}/\xi_{o}}+\frac{1/(2\xi_{o})}{\overline{\delta}_{o}-e_{o}/\xi_{o}}$
		as expressed in \eqref{eq:Attit_mu}. Recall $\tilde{\Omega}_{o}=\Omega-\mathcal{R}_{\tilde{Q}_{o}}^{\top}\hat{\Omega}$
		as in \eqref{eq:Attit_Omo_error}. From \eqref{eq:Attit_True_Dyn_Q},
		\eqref{eq:Attit_Obsv_Design}, and \eqref{eq:Attit_Ro_dot}, one finds
		\begin{align}
		& J\dot{\tilde{\Omega}}_{o}=J\dot{\Omega}-J\dot{\mathcal{R}}_{\tilde{Q}_{o}}^{\top}\hat{\Omega}-J\mathcal{R}_{\tilde{Q}_{o}}^{\top}\dot{\hat{\Omega}}\nonumber \\
		& =\left[J\Omega\right]_{\times}\Omega+(J[\tilde{\Omega}_{o}]_{\times}-[J\mathcal{R}_{\tilde{Q}_{o}}^{\top}\hat{\Omega}]_{\times})\mathcal{R}_{\tilde{Q}_{o}}^{\top}\hat{\Omega}-\tilde{\mathcal{R}}_{\tilde{Q}_{o}}^{\top}W_{\tau}\nonumber \\
		& =S(\Omega)\tilde{\Omega}_{o}-[J\tilde{\Omega}_{o}]_{\times}\tilde{\Omega}_{o}-\mathcal{R}_{\tilde{Q}_{o}}^{\top}W_{\tau}\label{eq:Attit_Omo_dot}
		\end{align}
		such that
		\begin{align}
		& \left[J\Omega\right]_{\times}\Omega+(J[\tilde{\Omega}_{o}]_{\times}-[J\mathcal{R}_{\tilde{Q}_{o}}^{\top}\hat{\Omega}]_{\times})\mathcal{R}_{\tilde{Q}_{o}}^{\top}\hat{\Omega}\nonumber \\
		& =([J\Omega]_{\times}-J[\Omega]_{\times}-[\Omega]_{\times}J)\tilde{\Omega}_{o}-[J\tilde{\Omega}_{o}]_{\times}\tilde{\Omega}_{o}\nonumber \\
		& =S(\Omega)\tilde{\Omega}_{o}-[J\tilde{\Omega}_{o}]_{\times}\tilde{\Omega}_{o}\label{eq:Attit_S_Omo}
		\end{align}
		where $S(\Omega)=[J\Omega]_{\times}-J[\Omega]_{\times}-[\Omega]_{\times}J$
		is a skew symmetric matrix. Consider the following Lyapunov function
		candidate
		\begin{equation}
		V_{o}=E_{o}^{2}+(1-\tilde{q}_{o0})+\frac{1}{2\gamma_{\Omega}}\tilde{\Omega}_{o}^{\top}J\tilde{\Omega}_{o}\label{eq:Attit_Vo}
		\end{equation}
		In view of \eqref{eq:Attit_RI2_dot}, \eqref{eq:Attit_Eo_dot}, and
		\eqref{eq:Attit_Omo_dot}, and with direct substitution of $W_{\Omega}$
		and $W_{\tau}$ by their definitions in \eqref{eq:Attit_Obsv_Design},
		one finds the time derivative of $V_{o}$ in \eqref{eq:Attit_Vo}
		as follows:
		\begin{align}
		\dot{V}_{o}= & -(E_{o}\Delta_{o}+1)\tilde{q}_{o}^{\top}(\tilde{\Omega}_{o}-\mathcal{R}_{\tilde{Q}_{o}}^{\top}W_{\Omega})-\frac{2E_{o}\Delta_{o}\dot{\xi}_{o}}{\xi_{o}}(1-\tilde{q}_{o0})\nonumber \\
		& +\frac{1}{\gamma_{o}}\tilde{\Omega}_{o}^{\top}(S(\Omega)\tilde{\Omega}_{o}-[J\tilde{\Omega}_{o}]_{\times}\tilde{\Omega}_{o}-\mathcal{R}_{\tilde{Q}_{o}}^{\top}W_{\tau})\nonumber \\
		\leq & -k_{o}(E_{o}^{2}\Delta_{o}^{2}+1)||\tilde{q}_{o}||^{2}\label{eq:Attit_Vo_dot}
		\end{align}
		where $1-\tilde{q}_{o0}\leq1-\tilde{q}_{o0}^{2}=||\tilde{q}_{o}||^{2}$,
		and since $\ell_{o}>\dot{\xi}_{o}/\xi_{o}$, $k_{o}$ is selected
		such that $k_{o}\geq\ell_{o}$. Note that $\tilde{\Omega}_{o}^{\top}S(\Omega)\tilde{\Omega}_{o}=0$
		and $\tilde{\Omega}_{o}^{\top}[J\tilde{\Omega}_{o}]_{\times}\tilde{\Omega}_{o}=0$.
		By the definition of $E_{o}$ in \eqref{eq:Attit_trans}, $\underline{\delta}_{o}=\overline{\delta}_{o}$
		implies that $E_{o}>0$ for all $1>|\tilde{q}_{o0}|$ and $E_{o}=0$
		only at $\tilde{q}_{o0}=\pm1$, see Remark \eqref{rem:Attit_PFF_e_E_0}.
		Hence, $V_{o}$ is a non-increasing function indicating that $E_{o}$,
		$e_{o}$, and $\tilde{\Omega}_{o}$ are bounded, $\lim_{t\rightarrow\infty}\tilde{q}_{o}=0_{3\times1}$,
		and $\lim_{t\rightarrow\infty}\mathcal{R}_{\tilde{Q}_{o}}=\mathbf{I}_{3}$.
		Thus, $W_{\Omega},W_{\tau}\rightarrow0_{3\times1}$, and, based on
		Barbalat Lemma, $\ddot{\tilde{Q}}_{o}$ and $\ddot{E}_{o}$ are bounded
		by $\dot{\tilde{Q}}_{o}\rightarrow0_{4\times1}$. From \eqref{eq:Attit_Qo_dot},
		one has $\tilde{\Omega}_{o}\rightarrow0_{3\times1}$, and, based on
		Barbalat Lemma, $\ddot{\tilde{\Omega}}_{o}$ is bounded and $\dot{\tilde{\Omega}}_{o}\rightarrow0_{3\times1}$
		completing the proof. Consider the cross term $-\frac{d}{dt}\frac{\tilde{q}_{o0}}{2\delta_{o}}\tilde{\Omega}_{o}^{\top}\tilde{q}_{o}$
		with the following derivative:
		\begin{align}
		& -\frac{d}{dt}\frac{\tilde{q}_{o0}}{2\delta_{o}}\tilde{\Omega}_{o}^{\top}\tilde{q}_{o}=-\frac{1}{2\delta_{o}}\dot{\tilde{\Omega}}_{o}^{\top}\tilde{q}_{o}-\frac{1}{2\delta_{o}}\tilde{\Omega}_{o}^{\top}\dot{\tilde{q}}_{o}-\frac{\dot{\tilde{q}}_{o0}}{2\delta_{o}}\tilde{\Omega}_{o}^{\top}\tilde{q}_{o}\nonumber \\
		& \hspace{1em}\leq-\frac{1}{4\delta_{o}}||\tilde{\Omega}_{o}||^{2}+\frac{c_{o1}+E_{o}\Delta_{o}k_{o}}{\delta_{o}}||\tilde{\Omega}_{o}||\,||\tilde{q}_{o}|||\nonumber \\
		& \hspace{2em}+\frac{c_{o2}}{\delta_{o}}||\tilde{q}_{o}||^{2}\label{eq:Attit_Cross_Term}
		\end{align}
		where $\bar{c}_{J}=\overline{\lambda}(J^{-1})$, $\eta_{\Omega}=\sup_{t\geq0}S(\Omega)$,
		$c_{o1}=\frac{\eta_{\Omega}\bar{c}_{J}}{2}+\frac{\eta_{\Omega_{o}}\bar{c}_{J}}{2}+\frac{\eta_{\Omega_{o}}}{2}$,
		and $c_{o2}=\frac{\gamma_{o}\bar{c}_{J}}{2}$. Let $c_{m}=\max\{c_{o1},k_{o}\}$,
		and consider the following Lyapunov function candidate:
		\begin{equation}
		\mathcal{L}_{o}=E_{o}^{2}+\underbrace{(1-\tilde{q}_{o0})+\frac{1}{2\gamma_{\Omega}}\tilde{\Omega}_{o}^{\top}J\tilde{\Omega}_{o}-\frac{\tilde{q}_{o0}}{2\delta_{o}}\tilde{\Omega}_{o}^{\top}\tilde{q}_{o}}_{\mathcal{L}_{qo}}\label{eq:Attit_Lo}
		\end{equation}
		with $e_{o}=[||\tilde{q}_{o0}||,||\tilde{q}_{o}||,||\tilde{\Omega}_{o}||]^{\top}$
		\begin{align*}
		& e_{o}^{\top}\underbrace{\left[\begin{array}{ccc}
			2 & 0 & -1\\
			0 & 2 & -\frac{1}{4\delta_{o}}\\
			-1- & \frac{1}{4\delta_{o}} & \frac{\underline{\lambda}_{J}}{2\gamma_{\Omega}}
			\end{array}\right]}_{P_{1}}e_{o}\leq\mathcal{L}_{qo}\\
		& \hspace{8em}\leq e_{o}^{\top}\underbrace{\left[\begin{array}{ccc}
			2 & 0 & 1\\
			0 & 2 & \frac{1}{4\delta_{o}}\\
			1 & \frac{1}{4\delta_{o}} & \frac{\underline{\lambda}_{J}}{2\gamma_{\Omega}}
			\end{array}\right]}_{P_{2}}e_{o}
		\end{align*}
		$P_{1}$ and $P_{2}$ can be made positive by selecting $\underline{\lambda}_{J}\frac{16\delta_{o}^{2}}{1+32\delta_{o}^{2}}>\gamma_{\Omega}$.
		From \eqref{eq:Attit_Vo_dot} and \eqref{eq:Attit_Cross_Term}, and
		selecting $\delta_{o}>\frac{2c_{o2}}{c_{\ell}}$, the time derivative
		of $\mathcal{L}_{o}$ in \eqref{eq:Attit_Lo} is
		\begin{align}
		\dot{\mathcal{L}}_{o}\leq & -\varepsilon_{o}^{\top}\underbrace{\left[\begin{array}{ccc}
			k_{o} & 0 & \frac{c_{m}}{2\delta_{o}}\\
			0 & k_{o} & \frac{c_{m}}{2\delta_{o}}\\
			\frac{c_{m}}{2\delta_{o}} & \frac{c_{m}}{2\delta_{o}} & \frac{1}{4\delta_{o}}
			\end{array}\right]}_{A_{o}}\varepsilon_{o}\label{eq:Attit_Lo_dot}
		\end{align}
		where $\varepsilon_{o}=\left[||\tilde{q}_{o}||,E_{o}\Delta_{o}||\tilde{q}_{o}||,||\tilde{\Omega}_{o}||\right]^{\top}$.
		$A_{o}$ can be made positive if $\delta_{o}>\frac{2c_{m}^{2}}{k_{o}}$
		such that
		\begin{align}
		\dot{\mathcal{L}}_{o}\leq & -\underline{\lambda}_{o}(E_{o}^{2}\Delta_{o}^{2}+1)||\tilde{q}_{o}||^{2}-\underline{\lambda}_{o}||\tilde{\Omega}_{o}||^{2}\label{eq:Attit_Lo_dot_Final}
		\end{align}
		with $\underline{\lambda}_{o}=\underline{\lambda}(A_{o})$ being the
		minimum singular value of $A_{o}$ proving Theorem \ref{thm:Theorem1_OBSV}.\end{proof}
	
	\subsection{Observer-based Controller with Guaranteed Convergence\label{subsec:Observer_Controller}}
	
	Consider the following controller design with control input $\tau$
	for \eqref{eq:Attit_True_Dyn_Q}: 
	\begin{equation}
	\begin{cases}
	\dot{Q}_{a} & =\frac{1}{2}Q_{a}\odot\overline{\beta_{a}}=\frac{1}{2}\Gamma(\beta_{a})Q_{a},\hspace{1em}Q_{a}(0)\in\mathbb{S}^{3}\\
	\beta_{a} & =-k_{\beta}(E_{a}\Delta_{a}+1)\mathcal{R}_{\tilde{Q}_{a}}^{\top}\tilde{q}_{a}\\
	W_{c} & =-k_{w}(E_{a}\Delta_{a}\tilde{q}_{a}+\tilde{q}_{c})\\
	\tau & =-W_{c}-k_{c}(\mathcal{R}_{\tilde{Q}_{o}}^{\top}\hat{\Omega}-\mathcal{R}_{\tilde{Q}_{c}}^{\top}\Omega_{d})\\
	& \hspace{1em}+[\mathcal{R}_{\tilde{Q}_{c}}^{\top}\Omega_{d}]_{\times}J\mathcal{R}_{\tilde{Q}_{c}}^{\top}\Omega_{d}+J\mathcal{R}_{\tilde{Q}_{c}}^{\top}\dot{\Omega}_{d}
	\end{cases}\label{eq:Attit_Cont_Design}
	\end{equation}
	where $\tilde{Q}_{c}=[\tilde{q}_{c0},\tilde{q}_{c}^{\top}]^{\top}=Q_{d}^{-1}\odot Q$
	is the unit-quaternion error, $Q_{d}$ is the desired unit-quaternion,
	$Q_{a}=[q_{a0},q_{a}^{\top}]^{\top}$ is the auxiliary unit-quaternion,
	$\tilde{Q}_{a}=[\tilde{q}_{a0},\tilde{q}_{a}^{\top}]^{\top}=Q_{a}^{-1}\odot\tilde{Q}_{c}$,
	$\mathcal{R}_{\tilde{Q}_{c}}$ is the attitude control error, $\mathcal{R}_{\tilde{Q}_{a}}$
	is the attitude auxiliary error, $e_{a}=1-\tilde{q}_{a0}$, $E_{a}=\frac{1}{2}\text{ln}\frac{\underline{\delta}_{a}+e_{a}/\xi_{a}}{\overline{\delta}_{a}-e_{a}/\xi_{a}}$
	is the transformed error, $W_{c}$ and $\beta_{a}$ are the correction
	factors, and $\dot{\Omega}_{d}$ is the derivative of the desired
	angular velocity. Additionally, $\xi_{a}$ is the PPF defined in \eqref{eq:Attit_Presc}
	with $\xi_{a}^{0}>e_{a}(0)$, and $k_{w}$, $k_{c}$, $k_{\beta}$,
	and $\underline{\delta}_{a}=\overline{\delta}_{a}>e_{a}(0)$ are positive
	constants.
	\begin{thm}
		\label{thm:Theorem2_CONT}Consider the dynamics in \eqref{eq:Attit_True_Dyn_Q}
		and the control law in \eqref{eq:Attit_Cont_Design}. Let Assumption
		\ref{Assum:Attit_desired_w} hold. Let the design parameters $k_{w}$,
		$k_{c}$, $k_{\beta}$, $\overline{\delta}_{a}=\underline{\delta}_{a}>e_{a}(0)$,
		$\xi_{a}^{0}>e_{a}(0)$, and $\xi_{a}^{\infty}$ be positive constants
		with $E_{a}(0)\in\mathcal{L}_{\infty}$. Then, 1) $E_{a}$, $e_{a}$,
		and $\Omega$ are globally bounded, and 2) starting from any initial
		conditions, all $E_{a}$, $e_{a}$, and $\tilde{\Omega}_{c}$ converge
		asymptotically to the origin with $\lim_{t\rightarrow\infty}\tilde{q}_{a}\rightarrow0_{3\times1}$,
		$\lim_{t\rightarrow\infty}\tilde{q}_{a0}\rightarrow\pm1$ $\lim_{t\rightarrow\infty}\tilde{q}_{c}\rightarrow0_{3\times1}$,
		and $\lim_{t\rightarrow\infty}\tilde{q}_{c0}\rightarrow\pm1$.
	\end{thm}
	\begin{proof}Consider $\tilde{Q}_{c}=Q_{d}^{-1}\odot Q$ as in \eqref{eq:Attit_Rc_error}.
		From \eqref{eq:Attit_True_Dyn_Q} and \eqref{eq:Attit_Rdot_desired},
		one finds
		\begin{align}
		\dot{\tilde{Q}}_{c} & =\frac{1}{2}\tilde{Q}_{c}\odot(\overline{\Omega}-\tilde{Q}_{c}^{-1}\odot\overline{\Omega_{d}}\odot\tilde{Q}_{c})\nonumber \\
		& =\frac{1}{2}\tilde{Q}_{c}\odot[0,\tilde{\Omega}_{c}^{\top}]^{\top}\label{eq:Attit_Qc_dot}\\
		\dot{\mathcal{R}}_{\tilde{Q}_{c}} & =\mathcal{R}_{\tilde{Q}_{c}}[\tilde{\Omega}_{c}]_{\times}\label{eq:Attit_Rc_dot}
		\end{align}
		Hence, in view of \eqref{eq:Attit_Qo_dot} and \eqref{eq:Attit_RI2_dot},
		one finds that $\dot{\tilde{q}}_{c0}=-\frac{1}{2}\tilde{q}_{c}^{\top}\tilde{\Omega}_{c}$
		and $\dot{\tilde{q}}_{c}=\frac{1}{2}(\tilde{q}_{c0}\mathbf{I}_{3}+[\tilde{q}_{c}]_{\times})\tilde{\Omega}_{c}$.
		In the same spirit, $\dot{\tilde{Q}}_{a}=\frac{1}{2}\tilde{Q}_{c}\odot[0,(\tilde{\Omega}_{c}-\mathcal{R}_{\tilde{Q}_{a}}^{\top}\beta_{a})^{\top}]^{\top}$
		such that $\dot{\mathcal{R}}_{\tilde{Q}_{a}}=\mathcal{R}_{\tilde{Q}_{a}}\left[\tilde{\Omega}_{c}-\mathcal{R}_{\tilde{Q}_{a}}^{\top}\beta_{a}\right]_{\times}$.
		Thus, the transformed error dynamics of $E_{a}$ become
		\begin{align}
		\dot{E}_{a} & =\Delta_{a}(-\frac{1}{2}\tilde{q}_{a}^{\top}(\tilde{\Omega}_{c}-\mathcal{R}_{\tilde{Q}_{a}}^{\top}\beta_{a})-\frac{\dot{\xi}_{a}}{\xi_{a}}(1-\tilde{q}_{a0}))\label{eq:Attit_Ea_dot}
		\end{align}
		where $\Delta_{a}=\frac{1/(2\xi_{a})}{\underline{\delta}_{a}+e_{a}/\xi_{a}}+\frac{1/(2\xi_{a})}{\overline{\delta}_{a}-e_{a}/\xi_{a}}$
		as specified in \eqref{eq:Attit_mu}. Let $\tilde{\Omega}_{c}=\Omega-\mathcal{R}_{\tilde{Q}_{c}}^{\top}\Omega_{d}$
		as defined in \eqref{eq:Attit_Omc_error}. In view of the steps in
		\eqref{eq:Attit_Omo_dot} and from \eqref{eq:Attit_True_Dyn_Q}, \eqref{eq:Attit_Cont_Design},
		and \eqref{eq:Attit_Rc_dot}, one has
		\begin{align}
		J\dot{\tilde{\Omega}}_{c}= & S(\Omega)\tilde{\Omega}_{c}-[J\tilde{\Omega}_{c}]_{\times}\tilde{\Omega}_{c}-W_{c}-k_{c}(\tilde{\Omega}_{c}-\tilde{\Omega}_{o})\label{eq:Attit_Omc_dot}
		\end{align}
		with $[J\Omega]_{\times}\Omega+J[\tilde{\Omega}_{c}]_{\times}\mathcal{R}_{\tilde{Q}_{c}}^{\top}\Omega_{d}+[\mathcal{R}_{\tilde{Q}_{c}}^{\top}\Omega_{d}]_{\times}J\mathcal{R}_{\tilde{Q}_{c}}^{\top}\Omega_{d}=S(\Omega)\tilde{\Omega}_{c}-[J\tilde{\Omega}_{c}]_{\times}\tilde{\Omega}_{c}$,
		$S(\Omega)=[J\Omega]_{\times}-J[\Omega]_{\times}-[\Omega]_{\times}J$
		being a skew symmetric matrix, see \eqref{eq:Attit_S_Omo}. Consider
		the following Lyapunov function candidate
		\begin{equation}
		\mathcal{L}_{c}=E_{a}^{2}+2(1-\tilde{q}_{c0})+\frac{1}{2k_{w}}\tilde{\Omega}_{c}^{\top}J\tilde{\Omega}_{c}\label{eq:Attit_Lc}
		\end{equation}
		From \eqref{eq:Attit_Qc_dot}, \eqref{eq:Attit_Ea_dot}, and \eqref{eq:Attit_Omc_dot},
		and replacing $W_{c}$ and $\beta_{a}$ with their definitions in
		\eqref{eq:Attit_Cont_Design} one finds the time derivative of $\mathcal{L}_{c}$
		in \eqref{eq:Attit_Lc} as follows:
		\begin{align}
		& \dot{\mathcal{L}}_{c}=-E_{a}\Delta_{a}\tilde{q}_{a}^{\top}(\tilde{\Omega}_{c}-\mathcal{R}_{\tilde{Q}_{a}}^{\top}\beta_{a})-\tilde{q}_{c}^{\top}\tilde{\Omega}_{c}\nonumber \\
		& \hspace{1em}-\frac{2E_{a}\Delta_{a}\dot{\xi}_{a}}{\xi_{a}}(1-\tilde{q}_{a0})\nonumber \\
		& \hspace{1em}+\frac{1}{k_{w}}\tilde{\Omega}_{c}^{\top}(S(\Omega)\tilde{\Omega}_{c}-[J\tilde{\Omega}_{c}]_{\times}\tilde{\Omega}_{c}-W_{c}-k_{c}(\tilde{\Omega}_{c}-\tilde{\Omega}_{o}))\nonumber \\
		& \leq-k_{\beta}E_{a}^{2}\Delta_{a}^{2}||\tilde{q}_{a}||^{2}-\frac{k_{c}}{k_{w}}||\tilde{\Omega}_{c}||^{2}+\frac{k_{c}}{k_{w}}||\tilde{\Omega}_{c}||\,||\tilde{\Omega}_{o}||\label{eq:Attit_Lc_dot-2}
		\end{align}
		where $2E_{a}\Delta_{a}\frac{\dot{\xi}_{a}}{\xi_{a}}(1-\tilde{q}_{a0})\leq2\ell_{a}E_{a}\Delta_{a}||\tilde{q}_{a}||^{2}$
		and $k_{\beta}$ is selected such that $k_{\beta}\geq2\ell_{a}$.
		According to the skew symmetric definition, $\tilde{\Omega}_{c}^{\top}S(\Omega)\tilde{\Omega}_{c}=0$
		and $\tilde{\Omega}_{c}^{\top}[J\tilde{\Omega}_{c}]_{\times}\tilde{\Omega}_{c}=0$.
		Let us combine \eqref{eq:Attit_Lo} and \eqref{eq:Attit_Lc} to obtain
		the following Lyapunov function candidate: 
		\begin{equation}
		\mathcal{L}_{T}=\mathcal{L}_{o}+\mathcal{L}_{c}\label{eq:Attit_LT}
		\end{equation}
		From \eqref{eq:Attit_Lo_dot_Final} and \eqref{eq:Attit_Lc_dot-2},
		one finds
		\begin{align}
		\dot{\mathcal{L}}_{T}\leq & -\underline{\lambda}_{o}(E_{o}^{2}\Delta_{o}^{2}+1)||\tilde{q}_{o}||^{2}-k_{\beta}E_{a}^{2}\Delta_{a}^{2}||\tilde{q}_{a}||^{2}\nonumber \\
		& -\varepsilon_{\Omega}^{\top}\underbrace{\left[\begin{array}{cc}
			\underline{\lambda}_{o} & 0.5k_{c}/k_{w}\\
			0.5k_{c}/k_{w} & k_{c}/k_{w}
			\end{array}\right]}_{A_{\Omega}}\varepsilon_{\Omega}\label{eq:Attit_LT_dot}
		\end{align}
		where $\varepsilon_{\Omega}=[||\tilde{\Omega}_{o}||,||\tilde{\Omega}_{c}||]^{\top}$.
		$A_{\Omega}$ can be made positive by selecting $\underline{\lambda}_{o}>\frac{k_{c}}{4k_{w}}$.
		Let $\underline{\lambda}_{c}=\underline{\lambda}(A_{\Omega})$ be
		the minimum eigenvalue of $A_{\Omega}$. One obtains 
		\begin{align}
		\dot{\mathcal{L}}_{T}\leq & -\underline{\lambda}_{o}(E_{o}^{2}\Delta_{o}^{2}+1)||\tilde{q}_{o}||^{2}-k_{\beta}E_{a}^{2}\Delta_{a}^{2}||\tilde{q}_{a}||^{2}-\underline{\lambda}_{c}||\tilde{\Omega}_{o}||\nonumber \\
		& -\underline{\lambda}_{c}||\tilde{\Omega}_{c}||\label{eq:Attit_LT_dot_Final}
		\end{align}
		The inequality in \eqref{eq:Attit_LT_dot_Final} shows that $\ddot{\mathcal{L}}_{T}$
		is bounded and $\dot{\mathcal{L}}_{T}$ goes to zero proving Theorem
		\ref{thm:Theorem2_CONT} in addition to $\tilde{q}_{a}\rightarrow0_{3\times1}$
		and $\tilde{\Omega}_{c}\rightarrow0_{3\times1}$. Since $\underline{\delta}_{a}=\overline{\delta}_{a}$,
		$E_{a}\neq0$ for all $e_{a}\neq0$ and $E_{a}=0$ only at $e_{a}=0$.
		Thus, $e_{a},E_{a}\rightarrow0$ which implies that $\tilde{q}_{a}\rightarrow0_{3\times1}$
		and $\tilde{q}_{a0}\rightarrow\pm1$. Based on Barbalat Lemma, $\ddot{\tilde{\Omega}}_{c}$
		and $\ddot{\tilde{Q}}_{a}$ are bounded with $\dot{\tilde{\Omega}}_{c}\rightarrow0_{3\times1}$
		and $\dot{\tilde{Q}}_{a}\rightarrow0_{3\times1}$. From \eqref{eq:Attit_Qc_dot},
		$\tilde{\Omega}_{c}\rightarrow0_{3\times1}$ indicates that $\dot{\tilde{Q}}_{c}\rightarrow0_{4\times1}$.
		Moreover, from \eqref{eq:Attit_Omc_dot}, $J\dot{\tilde{\Omega}}_{c}\rightarrow0_{3\times1}$
		shows that $W_{c}\rightarrow0_{3\times1}$ which shows that $\tilde{q}_{c}\rightarrow0_{3\times1}$
		and, in turn, $\tilde{q}_{c0}\rightarrow\pm1$ completing the proof.\end{proof}
	
	Let $\Delta t$ be a small sample time. Algorithm \ref{alg:Alg_Disc0}
	lists the complete implementation steps of a discrete form of the
	proposed quaternion observer-based controller with guaranteed performance.
	
	\begin{algorithm}
		\caption{\label{alg:Alg_Disc0}Unit-quaternion observer-based controller with
			systematic convergence}
		
		\textbf{Initialization}:
		\begin{enumerate}
			\item[{\footnotesize{}1:}] Set $\hat{Q}_{0}=\hat{Q}[0],Q_{a}[0]\in\mathbb{S}^{3}$, and $\hat{\Omega}_{0}=\hat{\Omega}[0],\tau_{0}=\tau[0]\in\mathbb{R}^{3}$
			\item[{\footnotesize{}2:}] Start with $k=0$ and select $k_{o}$, $\gamma_{o}$, $\xi_{o}^{\infty}$,
			$\xi_{o}^{0},\underline{\delta}_{o}=\overline{\delta}_{o}>1$, $k_{w}$,
			$k_{c}$, $\xi_{c}^{\infty}$, and $\xi_{c}^{0},\underline{\delta}_{c}=\overline{\delta}_{c}>1$
		\end{enumerate}
		\textbf{while }(1)\textbf{ do}
		\begin{enumerate}
			\item[{\footnotesize{}3:}] Use $\mathfrak{r}_{i}[k]$ and $\mathfrak{b}_{i}[k]$ in \eqref{eq:Attit_Vector_norm}
			to reconstruct $Q_{y}[k]$, visit \cite{hashim2020AtiitudeSurvey}
			\item[{\footnotesize{}4:}] $\tilde{Q}_{o}=\hat{Q}_{k}^{-1}\odot Q_{y}[k]$, $\tilde{Q}_{c}=Q_{d}^{-1}[k]\odot\hat{Q}_{k}$,
			and $\tilde{Q}_{a}=Q_{a}^{-1}[k]\odot\tilde{Q}_{c}$
			\item[{\footnotesize{}5:}] $\hat{\tau}=\mathcal{R}_{\tilde{Q}_{o}}\tau$
			\item[{\footnotesize{}6:}] $\xi_{\star}[k]=(\xi_{\star}^{0}-\xi_{\star}^{\infty})\exp(-\ell_{\star}k\Delta t)+\xi_{\star}^{\infty}$,
			\textcolor{blue}{/{*} $\star=o,a$ {*}/}
			\item[{\footnotesize{}7:}] $e_{\star}=1-\tilde{q}_{\star0}$, \textcolor{blue}{/{*} $\star=o,a$
				{*}/}
			\item[{\footnotesize{}8:}] \textbf{if} $\tilde{q}_{\star0}<0$ \textbf{then}\vspace{1mm}
			\item[{\footnotesize{}9:}] \hspace{0.5cm}\hspace{0.5cm}$e_{\star}=1+\tilde{q}_{\star0}$\vspace{1mm}
			\item[{\footnotesize{}10:}] \textbf{end if}
			\item[{\footnotesize{}11:}] \textbf{if} $e_{\star}>\xi_{\star}[k]$ \textbf{then}\vspace{1mm}
			\item[{\footnotesize{}12:}] \hspace{0.5cm}\hspace{0.5cm}$\xi_{\star}[k]=e_{\star}+\epsilon$,\hspace{0.5cm}
			/{*} $\epsilon$ is a small constant {*}/\vspace{1mm}
			\item[{\footnotesize{}13:}] \textbf{end if}
			\item[{\footnotesize{}14:}] $E_{\star}=\frac{1}{2}\text{ln}\frac{\overline{\delta}_{\star}+e_{\star}/\xi_{\star}}{\overline{\delta}_{\star}-e_{\star}/\xi_{\star}}$
			\item[{\footnotesize{}15:}] $\Delta_{\star}=\frac{1/(2\xi_{\star})}{\overline{\delta}_{\star}+e_{\star}/\xi_{\star}[k]}+\frac{1/(2\xi_{\star})}{\overline{\delta}_{\star}-e_{\star}/\xi_{\star}[k]}$
			\item[{\footnotesize{}16:}] $W_{\Omega}=-k_{o}(E_{o}\Delta_{o}+1)\mathcal{R}_{\tilde{Q}_{o}}\tilde{q}_{o}$
			\item[{\footnotesize{}17:}] $W_{\tau}=-\gamma_{o}(E_{o}\Delta_{o}+1)\mathcal{R}_{\tilde{Q}_{o}}\tilde{q}_{o}$
			\item[{\footnotesize{}18:}] $\hat{Q}_{k+1}=\exp(\frac{1}{2}\Gamma(\hat{\Omega}+W_{\Omega})\Delta t)\hat{Q}_{k}$
			\item[{\footnotesize{}19:}] $\hat{\Omega}_{k+1}=\hat{\Omega}_{k}$
			\item[] $\hspace{3em}+\Delta t\hat{J}^{-1}([\hat{J}\hat{\Omega}]_{\times}\hat{\Omega}+\hat{\tau}+\hat{J}[\hat{\Omega}]_{\times}W_{\Omega}+W_{\tau})$
			\item[{\footnotesize{}20:}] $\beta_{a}=-k_{\beta}(E_{a}\Delta_{a}+1)\mathcal{R}_{\tilde{Q}_{a}}^{\top}\tilde{q}_{a}$
			\item[{\footnotesize{}21:}] $Q_{a}[k+1]=\exp(\frac{1}{2}\Gamma(\beta_{a})\Delta t)Q_{a}[k]$
			\item[{\footnotesize{}22:}] $W_{c}=-k_{w}(E_{a}\Delta_{a}\tilde{q}_{a}+\tilde{q}_{c})$
			\item[{\footnotesize{}23:}] $\tau[k]=-W_{c}-k_{c}(\mathcal{R}_{\tilde{Q}_{o}}^{\top}\hat{\Omega}-\mathcal{R}_{\tilde{Q}_{c}}^{\top}\Omega_{d})$\vspace{0.1em}
			\item[] $\hspace{4em}+[\mathcal{R}_{\tilde{Q}_{c}}^{\top}\Omega_{d}]_{\times}J\mathcal{R}_{\tilde{Q}_{c}}^{\top}\Omega_{d}+J\mathcal{R}_{\tilde{Q}_{c}}^{\top}\dot{\Omega}_{d}$
			\item[{\footnotesize{}24:}] $k=k+1$
		\end{enumerate}
		\textbf{end while}
	\end{algorithm}

	\section{Numerical Results \label{sec:SE3_Simulations}}
	
	This section reveals the robustness of of the novel guaranteed performance
	discrete quaternion observer-based controller described in Algorithm
	\ref{alg:Alg_Disc0} at low sampling rate of 200 Hz. Let $r_{1}=\left[1,1.2,1.3\right]^{\top}$
	and $r_{2}=\left[0,0,1\right]^{\top}$ be two non-collinear inertial-frame
	vectors with body-frame measured values being corrupted by zero-mean
	noise with a standard deviation of $0.08$, visit \eqref{eq:Attit_Vect_True}.
	Based on Remark \ref{rem:Rem1_Attit_Rec}, $\mathfrak{r}_{i}=r_{i}/||r_{i}||$
	and $\mathfrak{b}_{i}=b_{i}/||b_{i}||$ for $i=1,2$ with $\mathfrak{r}_{3}=\mathfrak{r}_{1}\times\mathfrak{r}_{2}$
	and $\mathfrak{b}_{3}=\mathfrak{b}_{1}\times\mathfrak{b}_{2}$. Let
	$\Omega(0)=[0.2,0.3,0.3]^{\top}$ and its initial estimate be $\hat{\Omega}(0)=[0,0,0]^{\top}$.
	In order to account for a large initial error of the unit-quaternion
	between $Q[0]$ and $\hat{Q}[0]$ and between $Q[0]$ and $Q_{d}[0]$,
	consider $Q[0]=[0.0087,0.3906,0.1302,0.9113]^{\top}$, $Q_{d}[0]=Q_{a}[0]=\hat{Q}[0]=[1,0,0,0]^{\top}$.
	Assume that the rigid-body's inertia matrix is $J={\rm diag}(0.016,0.015,0.03)$.
	Let the time derivative of the desired angular velocity be
	\[
	\dot{\Omega}_{d}=\left[\begin{array}{c}
	0.03\sin(0.3t+\pi/4)\\
	0.05\sin(0.4t+\pi/3)\\
	0.02\sin(0.2t+\pi/2)
	\end{array}\right]\,{\rm rad/sec}^{2}
	\]
	Select the design parameters as follows: $\xi_{o}^{0}=\xi_{a}^{0}=\overline{\delta}_{o}=\underline{\delta}_{o}=\overline{\delta}_{a}=\underline{\delta}_{a}=1.7$,
	$\xi_{o}^{\infty}=\xi_{a}^{\infty}=0.05$, $\ell_{o}=\ell_{a}=1$,
	$k_{o}=10$, $k_{w}=1$, and $\gamma_{o}=k_{c}=k_{\beta}=0.1$.
	
	For simplicity in demonstration, Fig. \ref{fig:BFV} shows the third
	component of the body-frame measurements with respect to the true
	values. Despite high noise level in the body-frame measurements of
	an IMU module illustrated by Fig. \ref{fig:BFV}, Fig. \ref{fig:Quat}
	reveals impressive tracking capabilities in case of large initialization
	error and fast maneuvering. The robustness and fast adaptation of
	the proposed approach are confirmed in Fig. \ref{fig:Qe} where attitude
	tracking errors are successfully regulated to the desired equilibrium
	point $\mathbf{Q}_{{\rm I}}=[\pm1,0,0,0]$ and angular velocity tracking
	errors are regulated to the origin. Furthermore, Fig. \ref{fig:Om_Tau}
	demonstrates the boundedness of the control signal. Note that unit-quaternion
	is subject to non-uniqueness such that for $Q_{1}=-Q_{2}\in\mathbb{S}^{3}$,
	$\mathcal{R}_{Q_{1}}=\mathcal{R}_{Q_{2}}\in\mathbb{SO}\left(3\right)$.
	Only for plotting purposes, if $\tilde{q}_{c0}\rightarrow\pm1$ and
	$\tilde{q}_{o0}\rightarrow\mp1$, multiply $\hat{Q}$ by $-1$ to
	end with $\tilde{q}_{o0}\rightarrow\pm1$.
	
	\begin{figure}[h!]
		\centering{}\includegraphics[scale=0.31]{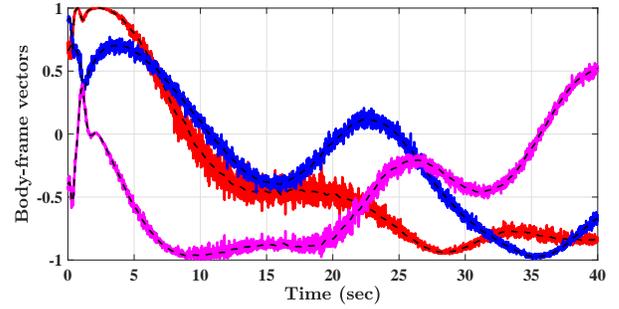}\caption{Body-frame vectors: Measurements (red, blue, and magenta solid-lines)
			vs true (black dashed-line).}
		\label{fig:BFV}
	\end{figure}

	\begin{figure}[h!]
		\centering{}\includegraphics[scale=0.29]{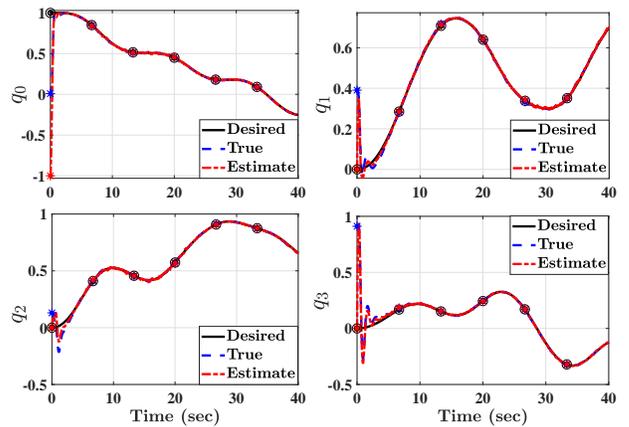}\caption{Unit-quaternion: desired $Q_{d}$, true $Q$, and estimated $\hat{Q}$.}
		\label{fig:Quat}
	\end{figure}
	
	\begin{figure}[h!]
		\centering{}\includegraphics[scale=0.32]{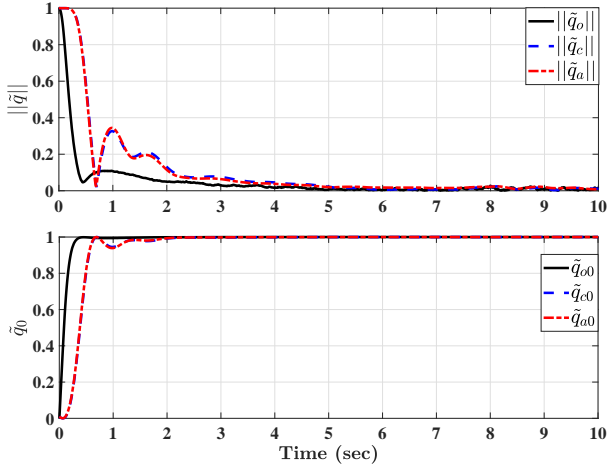}\caption{Errors in unit-quaternion: $\tilde{Q}_{o}$, $\tilde{Q}_{c}$, and
			$\tilde{Q}_{a}$.}
		\label{fig:Qe}
	\end{figure}
	
	\begin{figure}[h!]
		\centering{}\includegraphics[scale=0.3]{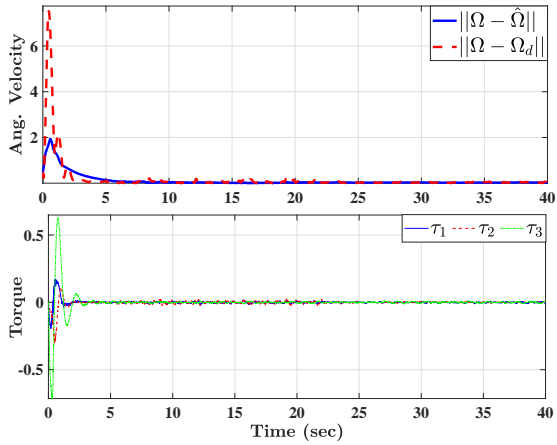}\caption{Control input (torque) and errors in angular velocity.}
		\label{fig:Om_Tau}
	\end{figure}

	\section{Conclusion \label{sec:SE3_Conclusion}}
	
	This paper addressed the challenge of velocity-free attitude tracking
	performed based solely on measurements obtained from low-cost inertial
	measurement units. A novel computationally cheap unit-quaternion observer-based
	controller ensuring almost global asymptotic stability of the overall
	closed loop signals has been proposed. Additionally, the transient
	and steady-state performance of the attitude tracking error has been
	shown to follow the dynamically reducing boundaries predefined by
	the user. Simulation results demonstrated high robustness and fast
	adaptation at a low sampling rate.
	
	\section*{Acknowledgment}
	
	The author would like to thank \textbf{Maria Shaposhnikova} for proofreading
	the article.

	\bibliographystyle{IEEEtran}
	\bibliography{bib_Attit_Cont}
\end{document}